\documentclass[12pt,a4paper]{amsart}

\usepackage{amsmath,amsfonts,amssymb}

\newcommand{\C}{\mathbb C}

\newcommand{\T}{\mathcal{T}}
\def\d{\partial}
\def\CP1{\mathbb{C}\mathrm{P}^1}
\def\un{{1\!\! 1}}
\def\r{\mathfrak{r}}
\def\s{\mathfrak{s}}

\newtheorem{theorem}{Theorem}

\newtheorem{lemma}[theorem]{Lemma}
\newtheorem{corollary}[theorem]{Corollary}

\theoremstyle{definition}

\title[On deformations of quasi-Miura transformations]{On deformations of quasi-Miura transformations and the Dubrovin-Zhang bracket}

\author{A. Buryak}

\address{A.~Buryak:\newline
Department of Mathematics,
University of Amsterdam, \newline
P.~O.~Box 94248, 1090 GE Amsterdam, 
The Netherlands\newline 
\indent and\newline
Department of Mathematics, Moscow State University,\newline
Leninskie gory, 119992 GSP-2 Moscow, Russia} 
\email{a.y.buryak@uva.nl, buryaksh@mail.ru}

\author{H. Posthuma}

\address{H.~Posthuma:\newline
Department of Mathematics,
University of Amsterdam, \newline
P.~O.~Box 94248, 1090 GE Amsterdam, 
The Netherlands}

\email{h.b.posthuma@uva.nl}

\author{S. Shadrin}

\address{S.~Shadrin:\newline
Department of Mathematics,
University of Amsterdam, \newline
P.~O.~Box 94248, 1090 GE Amsterdam, 
The Netherlands}
\email{s.shadrin@uva.nl}

\begin{document}

\begin{abstract}
In our recent paper we proved the polynomiality of a Poisson bracket for a class of infinite-di\-men\-sional
Hamiltonian systems of PDE's associated to semi-simple Frobenius structures. In the conformal (homogeneous) case, these systems are exactly the hierarchies of Dubrovin-Zhang, and the bracket is the first Poisson structure of their hierarchy.

Our approach was based on a very involved computation of a deformation formula for the bracket with respect to the Givental-Y.-P.~Lee Lie algebra action. In this paper, we discuss the structure of that deformation formula. In particular, we reprove it using a deformation formula for weak quasi-Miura transformation that relates our hierarchy of PDE's with its dispersionless limit. 
\end{abstract}

\maketitle

\tableofcontents

\section{Introduction}

In~\cite{DubZha2}, Dubrovin and Zhang proposed an axiomatic construction of a bi-Hamiltonian integrable hierarchy on the loop space of a confomal semi-simple Frobenius manifold. In~\cite{BurPosSha}, we have proposed a simple construction that associates a infinite Hamiltonian system of PDEs to an arbitrary cohomological field theory under some analytic assumptions on its partition function. In the conformal semi-simple case both constructions give the same system of PDEs, though it is clear that the second one is in some sense weaker since it requires more input and says nothing about integrability. Still it has a big advantage that it is applied in a much more general case.

One of the key concepts in both constructions is a weakened version of a Miura transformation, so-called quasi-Miura transformation (or even weak quasi-Miura transformation in non-homogeneous case), that is, a change of coordinates on the loop space. Namely, in order to obtain a Poisson bracket, we have to apply a change of coordinates to the fixed local operator $\d_x$. The weak quasi-Miura transformation, of course, depends on the particular cohomological field theory that we started with. 

The key property of that Poisson bracket, conjectured in~\cite{DubZha2} and proved in~\cite{BurPosSha}, is the fact that its expansion in a dispersion parameter $\hbar$ is polynomial in the derivatives of the dependent variables for all semi-simple Frobenius structures. It was proved using the known results for the $n$ copies of KdV hierarchy and an explicit formula for Givental-Lee infinitesimal deformations of the bracket. Let us recall that the Givental group acts transitively on the space of all semi-simple Frobenius structures in a fixed dimension $n$ and it allows to transfer the results established for a particular Frobenius structure to the whole orbit of its action.

The main trouble with the Poisson bracket in~\cite{BurPosSha} is that we were able to work with it only using the property that it is the unique operator that takes a given set of Hamiltonians into the equations of the hierarchy, together with a detailed study of the Hamiltonians and equations.
This very indirect approach gives a remarkably complicated formula for the Givental-Lee deformations of the bracket that, at the first glance, has no nice structure at all and is obtained through a sequence of really complicated calculations that just occassionaly give a concise answer.
Also this approach says nothing about the quasi-Miura transformations, and, therefore, gives no new tools to study the second bracket of Dubrovin and Zhang.

The purpose of this paper is to revisit the deformation formula for the bracket in~\cite{BurPosSha}. That is, we obtain a nice deformation formula for the weak quasi-Miura transformation and use it to give an alternative, more conceptual proof for the deformation formula of the bracket. Of course, the paper is therefore unavoidably very technical.

\subsection{Organization of the paper}
In Section~2 we recall all basic concepts that we need in this paper, that is, we introduce the necessary minimum of variational calculus, main properties of cohomological field theories, the construction of the Dubrovin-Zhang hierarchies, and Givental deformations of CohFTs. In Section 3 we compute the infinitesimal deformations of the quasi-Miura transformations with respect to the Lie algebra of the Givental group. In Section 4 we use the results of Section 3 in order to revisit the deformation formulas for the Poisson bracket of the Dubrovin-Zhang hierarchies that we first obtained in~\cite{BurPosSha}.

\subsection{Acknowledgement} A.~B. and S.~S. were supported by a Vidi grant of the Netherlands Organization fo Scientific Research.
A.~B. was also partially supported by the grants RFBR-10-01-00678, NSh-8462.2010.1, and the Moebius Contest Foundation for Young Scientists.





\section{Basic concepts}
In this section we briefly recall the setup of \cite{BurPosSha}. Throughout, we consider a vector space $V$ of dimension $s$ equipped with an inner product $\left<~,~\right>$. We fix an orthonormal basis $\{e_\alpha\}$, $\alpha=1,\dots,s$, and write $\un$ for the element $\sum_{\alpha=1}^s e_\alpha$.

\subsection{The variational calculus}
Let $B$ be some open ball in $V$.
We consider the formal jet space $\mathcal{J}^\infty(S^1,B)$ of maps from $S^1$ to $B$, and denote a 
typical element of it by $\mathbf{u}=(u_{\alpha,0},u_{\alpha,1},u_{\alpha,2},\ldots )$, where $u_{\alpha,k}:=\partial^ku_\alpha\slash\partial x^k$, $k\geq 1$. The Poisson bracket on this formal jet space is defined
on the space of functionals of the form
\begin{equation}
F(u)=\oint f(x,u_\alpha,u_{\alpha,1},\ldots)dx
\end{equation}
where $x\in S^1$ and 
\begin{equation}
 f(x,u_\alpha,u_{\alpha,1},\ldots)=\sum_{k=0}^\infty \hbar^k f_k(x,u_\alpha,u_{\alpha,1},\ldots)
\end{equation}
is a formal power series where in each degree $k$, the function $f_k$ depends smoothly on $x$ and a finite number $u_{\alpha,i}$. The nice case is when each $f_k$ is a polynomial in the variables $u_{\alpha,i},~i\geq 1$ of degree $2k$, where $\deg(u_{\alpha,i})=i$. 

The total derivative acting on such functionals is defined as
\begin{equation}
\partial_xf:=\frac{\partial f}{\partial x}+\sum_{\alpha,k}\frac{\partial f}{\partial u_{\alpha,k}}u_{\alpha, k+1},
\end{equation}
and $\oint\partial_xfdx=0$. The variational derivative is given by
\begin{equation}
\frac{\delta f}{\delta u_\alpha}:=\sum_{s=0}^\infty (-1)^s\partial^s_x\frac{\partial f}{\partial u_{\alpha,s}}.
\end{equation}

The natural group of coordinate transformations acting on this formal loop space is given by tranformations of the form
\begin{equation}
\label{miura-transf}
u_\alpha\mapsto \tilde{u}_\alpha:=\sum_{k=0}^\infty \hbar^k G_{\alpha,k}(u_\alpha,\ldots,u_{\alpha,n_k}),
\quad \det\left(\frac{\partial G_{\alpha,0}}{\partial u_\beta}\right)\neq 0.
\end{equation}
In case each $G_{\alpha,k}$ is a polynomial of degree $k$ in the variables $u_{\alpha,i},~i\geq 1$ (and therefore depends only on $u_{\alpha},u_{\alpha,1},\dots,u_{\alpha,2k}$), this is called a \emph{Miura transformation}. When $G_{\alpha,k}$ is merely a rational function of degree $2k$ (and therefore may depend on higher derivatives than $u_{\alpha,2k}$), it is called a \emph{quasi-Miura transformation}. We shall refer to the general case, where each $G_k$ is a function of a finite number of variables $(u_\alpha,\ldots,u_{\alpha,n_k})$, as a \emph{weak quasi-Miura transformation}.

The following Lemma is very useful in many of the calculations later on:
\begin{lemma} \label{lem:useful-lemma} For an arbitrary function $U$ and for any $\alpha=1,\dots,s$, we have:
\begin{equation}
\d_x \circ \sum_{s=0}^\infty \frac{\d U}{\d u_{\alpha,s}} \d_x^s = \sum_{s=0}^\infty \frac{\d \left(\d_xU\right)}{\d u_{\alpha,s}} \d_x^s
\end{equation}
\end{lemma}
\begin{proof} A straightforward calculation (we use only that $\d_x \circ \d/\d u_{\alpha,s} = \d/\d u_{\alpha,s} \circ \d_x -\d/\d u_{\alpha,s-1}$):
\begin{align}
&\d_x \circ \sum_{s=0}^\infty \frac{\d U}{\d u_{\alpha,s}} \d_x^s 
 = \sum_{s=0}^\infty \d_x\frac{\d U}{\d u_{\alpha,s}} \d_x^s + \sum_{s=0}^\infty \frac{\d U}{\d u_{\alpha,s}} \d_x^{s+1} \\ \notag
& = \sum_{s=0}^\infty \frac{\d \left(\d_xU\right)}{\d u_{\alpha,s}} \d_x^s -  \sum_{s=1}^\infty \frac{\d U}{\d u_{\alpha,s-1}} \d_x^{s}
+ \sum_{s=0}^\infty \frac{\d U}{\d u_{\alpha,s}} \d_x^{s+1} \\ \notag
& = \sum_{s=0}^\infty \frac{\d \left(\d_xU\right)}{\d u_{\alpha,s}} \d_x^s.
\end{align}
\end{proof}

For any weak quasi-miura transformation as in Equation~\eqref{miura-transf}, we introduce the differential operator 
\begin{equation}
L^\alpha_{\beta}=\sum_{s=0}^\infty \frac{\partial \tilde{u}_{\alpha}}{\partial u_{\beta,s}}\partial_x^s
\end{equation}
and its formal adjoint
\begin{equation}
\left(L^*\right)^\alpha_{\beta}=\sum_{s=0}^\infty (-\d_x)^s \circ \frac{\partial \tilde{u}_{\alpha}}{\partial u_{\beta,s}}
\end{equation}

Using the formal adjoint of the equality of Lemma \ref{lem:useful-lemma}, one easily shows that
\begin{equation}
\label{eq:L-var-der}
\sum_{\alpha}\left(L^*\right)^\beta_\alpha\circ\frac{\delta}{\delta \tilde{u}_\beta}=\frac{\delta}{\delta u_{\alpha}}.
\end{equation}

\subsection{Cohomological field theories}
We refer to \cite{KonMan,Sha,Tel} for an introduction to cohomological field theories. Here we are only interested in the associated partition function which is a formal power series
\begin{equation}
Z(t_0,t_1,\dots)=\exp\left(\sum_{g=0}^\infty\hbar^{g-1}F_g(t_0,t_1,\dots)\right),
\end{equation}
where $t_k=\{t_{1,k},\dots,t_{s,k}\}$, $k=0,1,2,\dots$, are the components of a vector $t=\sum_{\alpha,k}t_{\alpha,k}e_\alpha z^k$ of $V\otimes\C[[z]]$. We shall always assume that $\hbar\log Z$ is analytic in the variables $t_0$ and a formal power series in $\hbar$ and $t_{\alpha,k}$, $k\geq 1$.

The structure of a CohFT implies that $Z$ satsisfies the following two properties (see, e.~g.,~\cite{EguXio,KonMan2}):
\begin{itemize}
\item $Z$ is \emph{tame} in the sense that
\begin{align}\label{eq:TRR}
O_{\alpha,k}\left(\frac{\partial^2F_0}{\partial t_{\beta,p}\partial t_{\gamma, q}}\right)&=0, & & k>0,\ \alpha=1,\dots, s,\\ \notag
O_{\alpha,k}(F_g)&=0,& & g\geq 1,~ k>3g-2,~\alpha=1,\dots,s,
\end{align}
where the formal vector fields $O_{\alpha,k}$ are recursively defined in terms of $F_0$ by
\begin{align}
O_{\alpha,0}&:=\frac{\partial}{\partial t_{\alpha,0}}\\ \notag
O_{\alpha,k}&:=\frac{\partial}{\partial t_{\alpha,k}}-\sum_{i=0}^{k-1}\sum_\beta\frac{\partial^2F_0}{\partial t_{\alpha,i}\partial t_{\beta,0}}O_{\beta,k-i-1}.
\end{align}
Equations~\eqref{eq:TRR} are called \emph{the topological recursion relations}.
\item $Z$ satisfies the \emph{string equation}, that is, for each $g\geq 0$
\begin{equation}
\frac{\partial F_g}{\partial t_{\un,0}}=\sum_{\nu,k}t_{\nu,k+1}\frac{\partial F_g}{\partial t_{\nu,k}} +\frac{\delta_{g,0}}{2}\sum_{\alpha}t_{\alpha,0}^2.
\end{equation}
\end{itemize}
We do not impose any kind of homogeneity, in which case the cohomological field theory would be called \emph{conformal}. This means that the underlying Frobenius manifold is not required to have an Euler vector field. The reason for this is that homogeneity is not well compatible with the action of the Givental group, cf.~Section~\ref{sec:givental-group} below.

\subsection{The principal and full hierarchies}
Here we describe one of the core constructions of~\cite{DubZha2,BurPosSha}, the construction of two Hamiltonian hierarchies of PDE's associated to a CohFT. First, we introduce the notation 
\begin{equation}
\Omega_{\alpha,p;\beta,q}:=\sum_{g=0}^\infty\hbar^g\frac{\partial^2F_g}{\partial t_{\alpha,p}\partial t_{\beta, q}},
\end{equation}
and write $\Omega^{[0]}_{\alpha,p;\beta,q}$ for the constant term in $\hbar$ in this formal power series. We also 
introduce new coordinates
\begin{equation}
\label{w-coordinates}
w_{\alpha,k}:=\frac{\partial^k\Omega_{\alpha,0;\un,0}}{\partial t_{\un,0}^k}.
\end{equation}
Again, this is a formal power series in $\hbar$ and we write $w_{\alpha,k}=v_{\alpha,k}+O(\hbar)$. 
The $v_{\alpha,k}$ are the coordinates of the principal hierarchy that we now describe.

In genus zero ($g=0$), it follows from the topological recursion relation and the string equation that $\Omega^{[0]}_{\alpha,p;\beta,q}(t_0,t_1,\dots)=\Omega^{[0]}_{\alpha,p;\beta,q}(v,0,\dots)$ depends only on $v_\alpha=v_{\alpha,0}$. With this fact, the principal hierarchy is defined as
\begin{equation}
\frac{\partial v_\alpha}{\partial t_{\beta,q}}=\partial_x\left(\Omega^{[0]}_{\alpha,0;\beta,q}(v,0,0,\dots)\right),\quad \beta=1,\dots, s,~q\geq 0.
\end{equation}
It can be shown, cf.~\cite[Proposition 5]{BurPosSha}, that this is a Hamiltonian system of PDE's with the Hamiltonians given by $h^{[0]}_{\alpha,p}=\Omega^{[0]}_{\un,0;\alpha,p+1}$, and the Poisson bracket
\begin{equation}
\{F,G\}:=\oint\sum_\alpha\frac{\delta f}{\delta v_\alpha}\partial_x\frac{\delta g}{\delta v_\alpha}dx.
\end{equation}

For the full hierarchy, we have to use the coordinates $w_{\alpha,k}$. It follows from the higher genus topological recursion relations that the coordinate transformation
\begin{equation}
v_\alpha\mapsto w_\alpha=v_\alpha+\sum_{g=1}^\infty\hbar^g\frac{\partial^2 F_g}{\partial t_{\un,0}\partial t_{\alpha,0}},
\end{equation}
is a weak quasi-Miura transformation. That is, in each degree $k$ of $\hbar$, the $\Omega_{\alpha,p;\beta,q}$ depends only on a finite number of variables $\mathbf{w}=(w,w_1,\dots,$ $w_{3k})$. With this, the hierarchy associated to a CohFT is defined as
\begin{equation}
\frac{\partial w_\alpha}{\partial t_{\beta,q}}=\partial_x\left(\Omega_{\alpha,0;\beta,q}(w,w_1,w_2,\dots)\right)
\end{equation}
This system of PDE's comes with a tautological solution called the \emph{topological solution} and given by
$w_\alpha(x,t):=w_\alpha(t_0+x,t_1,\dots)$ where we shift along $t_{\un,0}$ with $x$, so that $w_{\alpha,k}=\partial_x^kw_\alpha$, consistent with \eqref{w-coordinates}.

To discuss the Hamiltonian structure of this hierarchy, we use the coordinate change $v_\alpha\mapsto w_\alpha$, which we implement using the operator
\begin{equation}
\label{def-L}
L^\alpha_{\beta}:=\sum_{s=0}^\infty\frac{\partial w_\alpha}{\partial v_{\beta,s}}\partial_x^s.
\end{equation}
With this operator, we define the Poisson bracket of the full hierarchy as
\begin{equation}
\{F,G\}:=\oint\sum_{\alpha,\beta,\gamma}\frac{\delta f}{\delta w_\alpha}L^\alpha_{\gamma}\partial_x\left(L^*\right)^\beta_{\gamma}\frac{\delta g}{\delta w_\beta}dx.
\end{equation}
The Hamiltonians are defined by $h_{\alpha,p}:=\Omega_{\alpha,p+1;\un,0}$. Since $\partial\slash\partial t_{\un,0}=\partial_x$, these are just the Hamiltonians of the principal hierarchy, deformed by an $\partial_x$-exact term. With this, and the equality \eqref{eq:L-var-der}, we compute:
\begin{align}
\{H_{\alpha,p},H_{\beta,q}\}&=\oint \sum_{\gamma,\delta,\nu}\frac{\delta h_{\alpha,p}}{\delta w_\gamma}L^\gamma_{\nu}\partial_x\left(L^*\right)^\delta_{\nu}\frac{\delta h_{\beta,q}}{\delta w_\delta}dx\\ \notag
&=\oint \sum_{\gamma,\delta,\nu}\left(L^*\right)^\gamma_{\nu}\frac{\delta h^{[0]}_{\alpha,p}}{\delta w_\gamma}
\partial_x\left(L^*\right)^\delta_{\nu}\frac{\delta h^{[0]}_{\beta,q}}{\delta w_\delta}dx\\ \notag
&=\oint\sum_\gamma\frac{\delta h^{[0]}_{\alpha,p}}{\delta v_\gamma}\partial_x \frac{\delta h^{[0]}_{\beta,q}}{\delta v_\gamma} dx=0,
\end{align}
where in the last line we have used the fact that the $h^{[0]}_{\alpha,p}$ are the Hamiltonians of the principal hierarchy.

Useful expositions of these constructions are also available in~\cite{Dub2,Ros}.

\subsection{The Givental group}
\label{sec:givental-group}
The action of the Givental group on cohomological field theories is best described in terms of its Lie algebra. This Lie algebra splits into two parts, the upper (resp., lower) triangular part $\mathfrak{g}_+$ (resp., $\mathfrak{g}_-$), defined as
\begin{equation}
\mathfrak{g}_\pm:=\left\{ u(z):=\sum_{k>0}u_kz^{\pm k},~u_k\in{\rm End}(V),~u(-z)^t+u(z)=0\right\}.
\end{equation}
Explicit formulas for the action can conveniently be given by defining $q_{\alpha,p}:=t_{\alpha,p}-\delta_{p,1}$, introducing $q_{\alpha,p}$ for $p<0$, and redefining
\begin{equation}
\label{eq:F-0}
\tilde{F}_0:=F_0+\sum_{i\geq 0}(-1)^i q_{\mu,i}q_{\mu,-1-i}.
\end{equation}
With this alteration of $F_0$, the same formula leads to a new partition function $\tilde{Z}$ and $\tilde F= \hbar \log \tilde Z$.

The action of the upper triangular part on this partition function is given by associating to $r_\ell z^\ell$, $\ell\geq 1$, the formal differential operator
\begin{align}\label{eq:r-hat}
\widehat{r_\ell}:= \sum_{\alpha,\beta} (r_\ell)_{\alpha,\beta} & \left( 
\frac{\hbar}{2}
\sum_{\begin{smallmatrix}
i+j=\ell-1\\ -\infty<i,j<\infty
\end{smallmatrix}}\frac{(-1)^{i+1}\partial^2}{\partial q_{\alpha,i}\partial q_{\beta,j}}
+\sum_{i\geq 0}\frac{q_{\alpha,-1-i}\partial}{\partial q_{\beta,\ell-1-i}} \right. \\ \notag &
+\frac{1}{2\hbar} 
\left. \sum_{\begin{smallmatrix}
i+j=\ell-1\\ i,j\geq 0
\end{smallmatrix}}
(-1)^{j+1} q_{\alpha,-1-i}q_{\beta,-1-j}
\right).
\end{align}
To see that this operator coincides with the one given in \cite{Lee1} is a small computation where one uses that the factor by which $F_0$ is shifted in \eqref{eq:F-0} allows to turn the operator of multiplication by $q_{\alpha,i}$, $i\geq 0$, into the derivation $(-1)^i\d/\d q_{\alpha,-1-i}$.

The action of the lower triangular part is given by the formal vector field
\begin{align}\label{eq:s-hat}
\widehat{s_\ell}:= \sum_{\alpha,\beta} (s_\ell)_{\alpha,\beta} \left( 
\frac{\hbar}{2}
\sum_{\begin{smallmatrix}
i+j=-\ell-1\\ -\infty<i,j<\infty
\end{smallmatrix}}\frac{(-1)^{i+1}\partial^2}{\partial q_{\alpha,i}\partial q_{\beta,j}}
+\sum_{i\geq 0}\frac{q_{\alpha,-1-i}\partial}{\partial q_{\beta,-\ell-1-i}} \right).
\end{align}
associated to the element $s_\ell z^{-\ell}$, $\ell\geq 1$.

We rewrite the operators $\widehat{r_\ell}$ and $\widehat{s_\ell}$ in this, a bit unusual way, in order to shorten the computations below (cf.~\cite[Remarks 8 and 10]{BurPosSha}). In fact, all but the first summands in Equations~\eqref{eq:r-hat} and~\eqref{eq:s-hat} obviously do not contibute to the deformation formulas below; formally it can be derived, for instance, as a corollary of Lemma~\ref{lem:U-q} in Section 3.

We refer to~\cite{Giv2,Giv3,Lee1,Sha} for different expositions and further details of the Givental theory.  





\section{Deformation of quasi-Miura transformation}\label{sec:r-def-L}

\subsection{The $r$-deformation formula}
In this section, we prove a formula for the $r$-deformation of the weak quasi-Miura transformation given by the operator
$L_{\beta}^{\alpha}:=\sum_{s=0}^\infty L^{\alpha}_{\beta,s}\d_x^s$, where $L_{\beta,s}^{\alpha}=\d w_{\alpha}/\d v_{\beta,s}$.

By $r[w].$, $r[v].$, $r[t].$ we denote the $r_\ell$-deformations in the corresponding coordinates. We use the same notation also for $s_\ell$-deformations

\begin{theorem} \label{thm:def-quasi-miura} We have:
\begin{align}\label{eq:def-quasi-miura}
r[w].\sum_{s=0}^\infty L^\alpha_{\beta,s}\d_x^s & =  
\sum_{n=0}^\infty \frac{\d(r[t].w_{\alpha})}{\d w_{\gamma,n}}\d_x^n
\circ
\sum_{s=0}^\infty L^\gamma_{\beta,s}\d_x^s \\ \notag
& -
\sum_{s=0}^\infty L^\alpha_{\gamma,s}\d_x^s
\circ
\sum_{n=0}^\infty \frac{\d(r[t].v_{\gamma})}{\d v_{\beta,n}}\d_x^n \\ \notag
& - \sum_{m,s=0}^\infty r[t].w_{\delta,m}\frac{\d L^\alpha_{\beta,s}}{\d w_{\delta,m}}\d_x^s.
\end{align}
\end{theorem}

\begin{proof} The proof is a straightforward computation. Indeed, the chain rule implies that
\begin{equation}\label{eq:chain-w-v}
r[w].L^{\alpha}_{\beta,s} =r[v].L^{\alpha}_{\beta,s}-\sum_{\delta,m}r[v].w_{\delta,m}\frac{\d L^{\alpha}_{\beta,s}}{\d w_{\delta,m}}.
\end{equation}
Using that $r[v].$ commutes with $\d/\d v_{\beta,s}$ and the chain rule 
\begin{equation}
r[v].w_{\delta,m}=r[t].w_{\delta,m}-\sum_{\nu,l}r[t].v_{\nu,l}\frac{\d w_{\delta,m}}{\d v_{\nu,l}}
\end{equation}
for any values of $(\delta,m)$, including $(\alpha,0)$, we see that the right hand side of Equation~\eqref{eq:chain-w-v} is equal to
\begin{align}
&\frac{\d}{\d v_{\beta,s}}\left(r[t].w_{\alpha}-\sum_{\gamma,n}r[t].v_{\gamma,n}\frac{\d w_{\alpha}}{\d v_{\gamma,n}}\right) \\ \notag
&-\sum_{\delta,m}\left(r[t].w_{\delta,m}-\sum_{\nu,l}r[t].v_{\nu,l}\frac{\d w_{\delta,m}}{\d v_{\nu,l}}\right)
\frac{\d}{\d w_{\delta,m}}\left(\frac{\d w_{\alpha}}{\d v_{\beta,s}}\right)
\end{align}

Note that 
\begin{equation}
\sum_{\delta,m}\frac{\d w_{\delta,m}}{\d v_{\nu,l}}\frac{\d}{\d w_{\delta,m}}\frac{\d w_{\alpha}}{\d v_{\beta,s}}=\frac{\d^2 w_{\alpha}}{\d v_{\beta,s}\d v_{\nu,l}}.
\end{equation} 
Therefore, we have
\begin{equation}\label{equation:deformation formula1}
r[w].L_{\beta,s}^{\alpha}=\frac{\d(r[t].w_{\alpha})}{\d v_{\beta,s}}-\sum_{\gamma,n}\frac{\d(r[t].v_{\gamma,n})}{\d v_{\beta,s}}\frac{\d w_{\alpha}}{\d v_{\gamma,n}}-\sum_{\delta,m}r[t].w_{\delta,m}\frac{\d}{\d w_{\delta,m}}\frac{\d w_{\alpha}}{\d v_{\beta,s}}
\end{equation}
Lemma~\ref{lem:useful-lemma} implies that this formula is equivalent to the Equation~\eqref{eq:def-quasi-miura}.
\end{proof}

\subsection{Finite dependence on variables $\mathbf{v}$ and $\mathbf{w}$ }

There is a problem with the deformation formula given by Equation~\eqref{eq:def-quasi-miura}. It is not obvious that under this deformation the operator is still well-defined on the class of functions depending on a finite number of variables $\mathbf{v}$ and $\mathbf{w}$ in every term of $h$-expansion and it is not obvious that under this deformation the coefficients of the operator still themselves depend on a finite number of variables $\mathbf{v}$ and $\mathbf{w}$. 

Indeed, 
\begin{align}
r[t].v_{\alpha}& =  \frac{1}{2}\sum_{\ell=1}^\infty\sum_{\begin{smallmatrix}i+j=l-1\\ \mu,\nu \end{smallmatrix}}(-1)^{i+1}\left(r_\ell\right)_{\mu\nu}
\d_x\frac{\d}{\d q_{\alpha,0}}\left(\frac{\d \tilde F_0}{\d q_{\mu,i}}\frac{\d \tilde F_0}{\d q_{\nu,j}}\right) \\ \notag 
& = \sum_{\ell=1}^\infty\sum_{\begin{smallmatrix}i+j=l-1\\ \mu,\nu \end{smallmatrix}}(-1)^{i+1}\left(r_\ell\right)_{\mu\nu}
\frac{\d \tilde F_0}{\d q_{\mu,i}} \frac{\d v_\alpha}{\d q_{\nu,j}} + \mbox{ good terms },
\end{align}
where by ``good'' terms we mean the summand that depends on a finite number of variables $v_{\beta,q}$. We recall that the function $\d \tilde F_0/\d q_{\mu,i}$ doesn't behave nicely in coordinates $\mathbf{v}$. In the same way, we have:
\begin{align}
r[t].w_{\alpha} = \sum_{\ell=1}^\infty\sum_{\begin{smallmatrix}i+j=l-1\\ \mu,\nu \end{smallmatrix}}(-1)^{i+1}\left(r_\ell\right)_{\mu\nu}
\frac{\d \tilde F_0}{\d q_{\mu,i}} \frac{\d w_\alpha}{\d q_{\nu,j}} + \mbox{ good terms }.
\end{align}

We can collect here all the ``bad'' terms that potentially cause the problems with finite dependance on $v$ and $w$:
\begin{align} \label{eq:infinite-terms}
& \sum_{\gamma,n,s} \frac{\d\left(\frac{\d \tilde F_0}{\d q_{\mu,i}}\frac{\d w_{\alpha}}{\d q_{\nu,j}}\right)}{\d w_{\gamma,n}}\d_x^n
\circ
 L^\gamma_{\beta,s}\d_x^s 
 -
\sum_{\gamma,n,s} L^\alpha_{\gamma,s}\d_x^s
\circ
\frac{\d\left(\frac{\d \tilde F_0}{\d q_{\mu,i}}\frac{\d v_{\gamma}}{\d q_{\nu,j}}\right)}{\d v_{\beta,n}}\d_x^n \\ \notag
& - \sum_{\delta,m,s} \frac{\d \tilde F_0}{\d q_{\mu,i}}\frac{\d w_{\delta,m}}{\d q_{\nu,j}}\frac{\d L^\alpha_{\beta,s}}{\d w_{\delta,m}}\d_x^s.
\end{align}

Before we'll take care of these ``bad'' terms, let us prove a more general statement about this kind of formulas. 
\begin{lemma} \label{lem:U-q} For any function $U$ and for any variable $q=q_{\nu,j}$, $\nu$ and $j$ are some arbitrary indices, we have:
\begin{align}\label{eq:U-q}
& \sum_{\gamma,n,s}  \frac{\d\left(U\frac{\d w_{\alpha}}{\d q}\right)}{\d w_{\gamma,n}}\d_x^n
\circ L^\gamma_{\beta,s}\d_x^s 
 -
\sum_{\gamma,n,s}  L^\alpha_{\gamma,s}\d_x^s
\circ
\frac{\d\left(U\frac{\d v_{\gamma}}{\d q}\right)}{\d v_{\beta,n}}\d_x^n \\ \notag
& - \sum_{\delta,m,s} U\frac{\d w_{\delta,m}}{\d q}\frac{\d L^\alpha_{\beta,s}}{\d w_{\delta,m}}\d_x^s
=-\sum_{\nu,r,s} \frac{\d w_\alpha}{\d v_{\nu,r}} 
\frac{\d \left(\sum_{i=1}^{r} \binom{r}{i} \d_x^iU \frac{\d v_{\nu,r-i}}{\d q}\right)}{\d v_{\beta,s}} \d_x^s.
\end{align}
\end{lemma}

\begin{proof} The proof is a straignforward computation. Indeed,
\begin{align}
& \sum_{\delta,m,s} U\frac{\d w_{\delta,m}}{\d q}\frac{\d L^\alpha_{\beta,s}}{\d w_{\delta,m}}\d_x^s
= \sum_{s} U \frac{\d L^\alpha_{\beta,s}}{\d q}\d_x^s
= \sum_{\delta,m,s} U\frac{\d }{\d q}\left(\frac{\d w_\alpha}{\d q_{\delta,m}}\frac{\d q_{\delta,m}}{\d v_{\beta,s}}\right)\d_x^s \\ \notag
& =
\sum_{\delta,m,s} \frac{\d \left(U\frac{\d w_\alpha}{\d q}\right)}{\d q_{\delta,m}}\frac{\d q_{\delta,m}}{\d v_{\beta,s}}\d_x^s
-
\sum_{\delta,m,s} \frac{\d w_\alpha}{\d q} \frac{\d U}{\d q_{\delta,m}}\frac{\d q_{\delta,m}}{\d v_{\beta,s}}\d_x^s \\ \notag
& - \sum_{\delta,m,s} \frac{\d w_\alpha}{\d q_{\delta,m}} \sum_{\nu,r,\lambda,p} \frac{\d q_{\delta,m}}{\d v_{\nu,r}} U\frac{\d}{\d q}\left(
\frac{\d v_{\nu,r}}{\d q_{\lambda,p}} \right) \frac{\d q_{\lambda,p}}{\d v_{\beta,s}} \d_x^s \\ \notag
& = \sum_{s} \frac{\d \left(U\frac{\d w_\alpha}{\d q}\right)}{\d v_{\beta,s}}\d_x^s
 - \sum_{\nu,r,s} \frac{\d w_\alpha}{\d v_{\nu,r}} 
\frac{\d \left(U\frac{\d v_{\nu,r}}{\d q}\right)}{\d v_{\beta,s}} \d_x^s.
\end{align}
Using Lemma~\ref{lem:useful-lemma}, we see that the first summand here is equal to the first summand on the left hand side of Equation~\eqref{eq:U-q}. Meanwhile, using Lemma~\ref{lem:useful-lemma} once again, we see that 
\begin{align}
& - \sum_{\nu,r,s} \frac{\d w_\alpha}{\d v_{\nu,r}} 
\frac{\d \left(U\frac{\d v_{\nu,r}}{\d q}\right)}{\d v_{\beta,s}} \d_x^s
=-\sum_{\gamma,n,s}  L^\alpha_{\gamma,s}\d_x^s
\circ
\frac{\d\left(U\frac{\d v_{\gamma}}{\d q}\right)}{\d v_{\beta,n}}\d_x^n \\ \notag
&+
\sum_{\nu,r,s} \frac{\d w_\alpha}{\d v_{\nu,r}} 
\frac{\d \left(\sum_{i=1}^{r} \binom{r}{i} \d_x^iU \frac{\d v_{\nu,r-i}}{\d q}\right)}{\d v_{\beta,s}} \d_x^s,
\end{align}
which is exactly the second summand on the left hand side of Equation~\eqref{eq:U-q} and the right hand side of that equation.
\end{proof}

\begin{corollary}\label{cor:finite} Formula~\eqref{eq:infinite-terms} is equal to 
\begin{equation}
-\sum_{\lambda,r,s} \frac{\d w_\alpha}{\d v_{\nu,r}} 
\frac{\d \left(\sum_{i=1}^{r} \binom{r}{i} \d_x^{i-1}\Omega^{[0]}_{\mu,i;\un,0} \frac{\d v_{\nu,r-i}}{\d q_{\nu,j}}\right)}{\d v_{\beta,s}} \d_x^s.
\end{equation} 
The coefficients of this operator depend on a finite number of the variables $\mathbf{v}$ and $\mathbf{w}$ in every term of the $\hbar$-expansion.
\end{corollary}

This proves, in particular, that the deformation formula~\eqref{eq:def-quasi-miura} preserves the class of operators that are well-defined on the space of functions depending on a finite number of variables $\mathbf{v}$ and $\mathbf{w}$ in each term of the $\hbar$-expansion.


\subsection{The $s$-deformation formula}
In this section, we prove a formula for the $s$-deformation of the weak quasi-Miura transformation
$L_{\beta}^{\alpha}:=\sum_{s=0}^\infty L^{\alpha}_{\beta,s}\d_x^s$, where $L_{\beta,s}^{\alpha}=\d w_{\alpha}/\d v_{\beta,s}$.

\begin{theorem} \label{thm:s-def-quasi-miura} We have:
\begin{align}\label{eq:s-def-quasi-miura}
s[w].\sum_{s=0}^\infty L^\alpha_{\beta,s}\d_x^s & =  -\sum_{\delta}(s_1)_{\delta,\un} \frac{\d}{\d w_{\delta,0}} \sum_{s=0}^\infty L^\alpha_{\beta,s}\d_x^s.
\end{align}
\end{theorem}

\begin{proof} The proof is again a straightforward computation. As in the case of $r$-deformation, we use the chain rule:
\begin{equation}\label{eq:s-chain-w-v}
s[w].L^{\alpha}_{\beta,s} =s[v].L^{\alpha}_{\beta,s}-\sum_{\delta,m}s[v].w_{\delta,m}\frac{\d L^{\alpha}_{\beta,s}}{\d w_{\delta,m}}.
\end{equation}
One can easily compute that $s[v].w_\mu=(s_1)_{\mu,\un}-\sum_{\nu}(s_1)_{\nu,\un} \d w_{\mu}/\d v_{\nu,0}$.
Using this and the fact that $s[v].$ commutes with $\d/\d v_{\beta,s}$, we obtain 
\begin{align}
s[v].L^{\alpha}_{\beta,s} & = -\sum_{\nu}(s_1)_{\nu,\un} \frac{\d}{\d v_{\nu,0}} L^{\alpha}_{\beta,s}
\\ \notag
\sum_{\delta,m}s[v].w_{\delta,m}\frac{\d L^{\alpha}_{\beta,s}}{\d w_{\delta,m}}
& = \sum_{\delta} (s_1)_{\delta,\un} \frac{\d L^{\alpha}_{\beta,s}}{\d w_{\delta,0}}
- \sum_{\nu,\delta,m}(s_1)_{\nu,\un}\frac{\d w_{\delta,m}}{\d v_{\nu,0}}\frac{\d L^{\alpha}_{\beta,s}}{\d w_{\delta,m}}
\end{align}
The difference of these two expressions is equal to the right hand side of Equation~\eqref{eq:s-def-quasi-miura}.
\end{proof}




\section{The deformation formula}

Let us recall the deformation formula for the bracket. We use the following notations:
\begin{align}
\delta_{\xi}& := \frac{\delta}{\delta w_\xi} = \sum_{n=0}^\infty (-\d_x)^n\circ \frac{\d}{\d w_{\xi,n}};
\\ \notag
\T_{\xi,k} & := \sum_{n=0}^\infty \binom{n}{k} (-\d_x)^{n-k} \circ\frac{\d}{\d w_{\xi,n}}.
\end{align}
We use the convention that $\binom{n}{k}=0$ if $n\geq 0$ and $k<0$ or $k>n$.

In the case of $r$-deformations, we have in~\cite{BurPosSha} the following formula:
\begin{align}\label{eq:r-def-bracket}
& \sum_{s=1}^{\infty} \left(r[w]. A_{s}^{\beta\xi}\right) \d_x^s 
= \sum_{i+j=\ell-1} (-1)^{i+1} (\r_\ell)_{\mu\nu} \left[ \phantom{\sum_{s\ge 1,\xi}^\infty}\right. \\
\tag{I} &  \Omega_{\un,0;\nu,j} \sum_{\gamma,n} \frac{\d \Omega_{\mu,i;\beta,0}}{\d w_{\gamma,n}} \d_x^{n}\circ 
\sum_{s\ge 1,\xi}^\infty A^{\gamma,\xi}_s \d_x^s 
\\ & \tag{II} 
 - \sum_{\gamma,n} \sum_{a+b=n} \binom{n+1}{a}\d_x^b \Omega_{\un,0;\nu,j} \d_x^a \Omega_{\mu,i;\gamma,0}  
\sum_{s\ge 1,\xi} \frac{\d A^{\beta,\xi}_s}{\d w_{\gamma,n}} \d_x^s  
\\ &  \tag{III} 
+ \sum_{s\ge 1,\gamma} A^{\beta,\gamma}_s \sum_{f+e=s-1} \d_x^f \circ \Omega_{\un,0;\nu,j}
\d_x^e \circ \sum_{n=0}^\infty \T_{\gamma,n}\Omega_{\mu,i;\xi,0} (-\d_x)^{n+1}
\end{align}
\begin{align}
&  \tag{IV}
+ \Omega_{\beta,0;\nu,j}\sum_{\gamma,n}\frac{\d \Omega_{\mu,i;\un,0}}{\d w_{\gamma,n}} \d_x^n\circ \sum_{s\ge 1,\xi} A^{\gamma,\xi}_s \d_x^s 
\\ &  \notag
+ \sum_{s\ge 1,\gamma} A^{\beta\gamma}_s \d_x^s\circ 
\sum_{0\leq u\leq v} \left( \phantom{\binom{v}{u}} \right. \notag 
\\ &  \tag{V}
\binom{v}{u}
\T_{\gamma,v+1}\Omega_{\un,0;\nu,j} (-\d_x)^{v-u} \Omega_{\mu,i;\xi,0} (-\d_x)^{u+1}   
 \\ &   \tag{VI}
- 
\left. 
\binom{v+1}{u}
(-\d_x)^{v-u}\Omega_{\un,0;\nu,j} \T_{\gamma,v+1} \Omega_{\mu,i;\xi,0} (-\d_x)^{u+1}  
\right)
\\ 
&  \tag{VII} 
+ \sum_{s\ge 1,\gamma} A^{\beta\gamma}_s \sum_{e+f=s-1} \binom{s}{e} \d_x^e 
\delta_\gamma \Omega_{\un,0;\nu,j} \d_x^f \circ \Omega_{\mu,i;\xi,0} \d_x 
\\ &   \tag{VIII} 
-\d_x\Omega_{\beta,0;\nu,j-1}
\sum_{\gamma,m} \sum_{u=0}^{m-1} 
(-\d_x)^u \d_{\gamma,m} \Omega_{\mu,i+1;\un,0} \d_x^{m-1-u}\circ \sum_{s=1}^\infty A^{\gamma\xi}_s\d_x^s  
\\ &  \tag{IX}
-\d_x\Omega_{\beta,0;\nu,j-1}  \sum_{\gamma}\sum_{2\leq f \leq s}^\infty 
(-\d_x)^{s-f} \left( A^{\gamma\xi}_s \delta_{\gamma} \Omega_{\mu,i+1;\un,0}\right) \d_x^{f-1}
\\ 
&  \tag{X} 
+\frac{\hbar}{2} \left( 
\d_x \circ 
\sum_{\gamma, n} \frac{\d \Omega_{\beta,0;\mu,i;\nu,j}}{\d w_{\gamma,n}}
\d_x^{n} 
\circ \sum_{s\ge 1}  A^{\gamma,\xi}_s \d_x^s \right.
\\ 
&  \tag{XI}
+
\sum_{s\ge 1,\gamma}  A_{s}^{\beta\gamma} \d_x^s 
\circ 
\sum_{m=0}^\infty 
\T_{\gamma,m} \Omega_{\xi,0;\mu,i;\nu,j}
 (-\d_x)^{m+1}  
 \\ & \tag{XII}
\left.\left. - \sum_{n=0}^\infty\sum_{\zeta} \d_x^{n+1} 
\Omega_{\zeta,0;\mu,i;\nu,j} \sum_{s\ge 1} \frac{\d A_s^{\beta\xi}}{\d w_{\zeta,n}} 
\d_x^s\right)\right],
\end{align}
$\ell=1,2,3,\dots$. In the case of $s$-deformations, the formula looks like 
\begin{align}\label{eq:s-deformation-bracket}
& \sum_{s=1}^{\infty} \left(\sum_{\ell=1}^\infty s[w]. A_{s}^{\beta\xi}\right) \d_x^s 
= -\sum_{s=1}^{\infty} \left(\sum_{\gamma} (\s_1)_{\gamma,\un} \frac{\d A_{s}^{\beta\xi}}{\d w_{\gamma,0}}\right) \d_x^s.
\end{align}

We are going to reprove these formulas using Theorems~\ref{thm:def-quasi-miura} and~\ref{thm:s-def-quasi-miura}.

\subsection{The $r$-deformation formula}

\begin{theorem}\label{thm:r-def-bracket} Let $L$ be the operator of quasi-Miura tranformation discussed in Section~\ref{sec:r-def-L}. Then the 
deformation formula 
\begin{equation}\label{eq:r-def-bracket-L}
\sum_\alpha r[w].L^\beta_\alpha \circ \d_x \circ \left(L^*\right)^\xi_\alpha 
+ \sum_\alpha L^\beta_\alpha \circ \d_x \circ \left(r[w].L^*\right)^\xi_\alpha   
\end{equation}
is equal to the right hand side of Equation~\eqref{eq:r-def-bracket}.
\end{theorem}

The proof of this theorem is based on several easy observations about the structure of the deformation formula~\eqref{eq:def-quasi-miura}. It is just a sequence of very straightforward computations in Sections~\ref{sec:internalr}-~\ref{sec:proofthmr}, where the origin and the meaning of each term in Equation~~\eqref{eq:r-def-bracket} becomes completely clear. 

In order to simplify the exposition of our calculations, we split the formula~\eqref{eq:r-def-bracket-L} into several summands that are not well-defined as operators actions on a class of functions whose coefficients depend on an finite number of variables $\mathbf{v}$ and $\mathbf{w}$. We work with these summands independently. So, according to Corollary~\ref{cor:finite}, only the total sum of all our operators (whatever they turn out to be in our computations) is a well-defined operator whose coefficients depend on a finite numbers of variables $\mathbf{v}$ and $\mathbf{w}$ in each term of $\hbar$-expansion.

\subsection{Internal terms}\label{sec:internalr}
Let us begin with the following two terms that come from the second summand on the right hand side of Equation~\eqref{eq:def-quasi-miura}:
\begin{align}
& L^\beta_\alpha\circ \sum_{p\geq 0}\frac{\d r[t].v_{\alpha,0}}{\d v_{\omega,p}}\circ \d_x^{p}\circ \d_x \circ \left(L^*\right)^\xi_\omega
\\ \notag
& + L^\beta_\alpha\circ \d_x\circ \sum_{p\geq 0} (-\d_x)^p \circ \frac{\d r[t].v_{\omega,0}}{\d v_{\alpha,p}} \circ \left(L^*\right)^\xi_\omega
\end{align}
We call these terms ``internal'' since they have a form of the same quasi-Miura transformation of a deformed operator that depends only on genus $0$ data. This sum is equal to zero because of the following lemma.

\begin{lemma}\label{lem:r-1}
We have:
\begin{equation}\label{eq:int-terms}
\sum_{p\geq 0}\frac{\d r[t].v_{\alpha,0}}{\d v_{\beta,p}}\circ \d_x^{p}\circ \d_x
+ \d_x\circ \sum_{p\geq 0} (-\d_x)^p \circ \frac{\d r[t].v_{\beta,0}}{\d v_{\alpha,p}} =0.
\end{equation}
\end{lemma}

\begin{proof} We use that 
\begin{equation}
r[t].v_{\alpha,0}=\frac{1}{2}\sum_{l\geq 1}\sum_{i+j=l-1}(-1)^{i+1}(r_l)_{\mu\nu} \d_x\frac{\d}{\d q_{\alpha,0}}\left(\frac{\d \tilde F_0}{\d q_{\mu,i}}\frac{\d \tilde F_0}{\d q_{\nu,j}}\right)
\end{equation}
(and below we usually omit $\sum_{l\geq 1}\sum_{i+j=l-1}(-1)^{i+1}(r_l)_{\mu\nu}$ for briefness). Using Lemma~\ref{lem:useful-lemma} and the topological recursion relation we have:
\begin{align}
& \sum_{p\geq 0}\frac{\d r[t].v_{\alpha,0}}{\d v_{\beta,p}}\d_x^{p}
=\d_x \circ \sum_{p\geq 0}\frac{\d \left(\frac{\d \tilde F_0}{\d q_{\mu,i}}\Omega^{[0]}_{\alpha,0;\nu,j}\right)}{\d v_{\beta,p}}\d_x^{p}
\\ \notag &
=\d_x \circ \Omega^{[0]}_{\alpha,0;\nu,j} \sum_{p\geq 0}\frac{\d \left(\frac{\d \tilde F_0}{\d q_{\mu,i}}\right)}{\d v_{\beta,p}} \d_x^{p}
+\d_x \circ \frac{\d \tilde F_0}{\d q_{\mu,i}} \sum_{p\geq 0}\frac{\d \Omega^{[0]}_{\alpha,0;\nu,j} }{\d v_{\beta,p}} \d_x^{p}
\\ \notag &
=\d_x \circ \Omega^{[0]}_{\alpha,0;\nu,j}\d_x^{-1} \circ \sum_{p\geq 0}\frac{\d \Omega^{[0]}_{\un,0;\mu,i}}{\d v_{\beta,p}}\d_x^{p}
+\d_x \circ \frac{\d\tilde F_0}{\d q_{\mu,i}} \frac{\d \Omega^{[0]}_{\alpha,0;\nu,j}}{\d v_{\beta,0}}
\\ \notag &
=\d_x \circ \Omega^{[0]}_{\alpha,0;\nu,j} \d_x^{-1} \circ \frac{\d \Omega^{[0]}_{\un,0;\mu,i}}{\d v_{\beta,0}}
+\d_x \circ \frac{\d\tilde F_0}{\d q_{\mu,i}} \frac{\d \Omega^{[0]}_{\alpha,0;\nu,j}}{\d v_{\beta,0}}
\\ \notag &
=\d_x \circ \Omega^{[0]}_{\alpha,0;\nu,j}\d_x^{-1} \circ \Omega^{[0]}_{\beta,0;\mu,i-1}
+\d_x \circ \frac{\d\tilde F_0}{\d q_{\mu,i}} \frac{\d \Omega^{[0]}_{\alpha,0;\nu,j}}{\d v_{\beta,0}}.
\end{align}
In the same way,
\begin{align}
& \sum_{p\geq 0} (-\d_x)^p \circ \frac{\d r[t].v_{\beta,0}}{\d v_{\alpha,p}}
\\ \notag &
=\Omega^{[0]}_{\alpha,0;\nu,j-1}\d_x^{-1} \circ \Omega^{[0]}_{\beta,0;\mu,i}\d_x
-\frac{\d\tilde F_0}{\d q_{\mu,i}} \frac{\d \Omega^{[0]}_{\beta,0;\nu,j}}{\d v_{\alpha,0}}\d_x
\end{align}
Therefore the left hand side of Equation~\eqref{eq:int-terms} is equal to
\begin{align}
& \d_x \circ \left( \frac{\d \tilde F_0}{\d q_{\mu,i}} \frac{\d \Omega^{[0]}_{\alpha,0;\nu,j}}{\d v_{\beta,0}} - \frac{\d \tilde  F_0}{\d q_{\mu,i}} \frac{\d \Omega^{[0]}_{\beta,0;\nu,j}}{\d v_{\alpha,0}}\right)\circ \d_x
\\ \notag &
+ \d_x \circ \Omega^{[0]}_{\alpha,0;\nu,j}\circ \d_x^{-1} \circ \Omega^{[0]}_{\beta,0;\mu,i-1} \circ \d_x
\\ \notag &
+ \d_x \circ \Omega^{[0]}_{\alpha,0;\nu,j-1}\circ \d_x^{-1} \circ \Omega^{[0]}_{\beta,0;\mu,i} \circ \d_x
\end{align}
Here the first line is equal to zero just because
\begin{equation}\label{eq:triomega}
\frac{\d \Omega^{[0]}_{\alpha,0;\mu,i}}{\d v_{\beta,0}} = \frac{\d \Omega^{[0]}_{\beta,0;\mu,i}}{\d v_{\alpha,0}} = 
\left. \frac{\d \tilde F_0}{\d q_{\mu,i} \d q_{\alpha,0} \d q_{\beta,0}}\right| _{q_{\gamma,0}=v_{\gamma,0}, q_{\gamma,\geq 1} = 0}.
\end{equation}
The second and the third lines have exactly the same summands when we take the sum over all $i+j=l-1$, but one of them comes with the sign $(-1)^{i+1}$ while the other one appears with the sign $(-1)^{(i-1)+1}$. Therefore, they cancel each other when we take the sum over all $i+j=l-1$.
\end{proof}

\subsection{External terms} In this section we compute the two terms coming from the first summand on the right hand side of Equation~\eqref{eq:def-quasi-miura}. In terms of the operator $A=\sum_{s\geq 1} A^{\beta\xi}_s \d_x^s = \sum_{\alpha} L^\beta_\alpha\circ \d_x\circ \left(L^*\right)^\xi_\alpha$ they can be written as
\begin{equation}\label{eq:external-terms}
\sum_{\gamma,n} \frac{\d(r[t].w_{\beta})}{\d w_{\gamma,n}}\d_x^n \circ \sum_s A^{\gamma\xi}_s \d_x^s
+\sum_{\gamma,s} A^{\beta\gamma}_s \d_x^s \circ \sum_{\gamma,n} (-\d_x)^n\circ \frac{\d(r[t].w_{\xi})}{\d w_{\gamma,n}}
\end{equation}

\begin{lemma}\label{lem:r-2} Formula~\eqref{eq:external-terms} is equivalent to the sum of the lines~(I), (III), (IV), (V), (VI), (VII), (VIII), (IX), (X), and (XI) on the right hand side of Equation~\eqref{eq:r-def-bracket} and the following extra terms:
\begin{align}\label{eq:extra-terms}
& \sum_{l\geq 1}\sum_{i+j=l-1}(-1)^{i+1}(r_l)_{\mu\nu}
 \frac{\d \tilde F}{\d q_{\mu,i}} \cdot \\ \notag &
 \left(\sum_{\gamma,s,n} \frac{\d \left(\frac{\d w_\beta}{\d q_{\nu,j}}\right)}{\d w_{\gamma,n}}\d_x^n \circ  
A^{\gamma\xi}_s \d_x^s  + 
 \sum_{\gamma,s,n} A^{\beta\gamma}_s \d_x^s \circ (-\d_x)^n\circ 
 \frac{\d \left(\frac{\d w_\xi}{\d q_{\nu,j}}\right)}{\d w_{\gamma,n}}\right). 
\end{align}
\end{lemma}

\begin{proof}We use that 
\begin{align}
& r[t].w_{\alpha,0} \\ \notag &
=\frac{1}{2}\sum_{l\geq 1}\sum_{i+j=l-1}(-1)^{i+1}(r_l)_{\mu\nu} \frac{\d^2}{\d x \d q_{\alpha,0}}\left(\frac{\d \tilde F}{\d q_{\mu,i}}\frac{\d \tilde F}{\d q_{\nu,j}} +\hbar \frac{\d^2 \tilde F}{\d q_{\mu,i}\d q_{\nu,j}} \right)
\end{align}
(and below we usually omit $\sum_{l\geq 1}\sum_{i+j=l-1}(-1)^{i+1}(r_l)_{\mu\nu}$ for briefness). Using Lemma~\ref{lem:useful-lemma} we see that \begin{align}
\sum_{n} \frac{\d(r[t].w_{\beta})}{\d w_{\gamma,n}}\d_x^n
& = \d_x \circ \sum_{n} \frac{\d}{\d w_{\gamma,n}}\left(\frac{\d \tilde F}{\d q_{\mu,i}}\Omega_{\nu,j;\beta,0}\right)\d_x^n \\ \notag
& + \frac{\hbar}{2} \d_x \circ \sum_{n} \frac{\d \Omega_{\mu,i;\nu,j;\beta,0}}{\d w_{\gamma,n}}\d_x^n
\end{align}
The composition of the second summand here and the operator $A$ gives exactly the line~(X) in Equation~\eqref{eq:r-def-bracket}.
The composition of the first summand here and the operator $A$ is equal to
\begin{align}
& \d_x \circ \sum_{n} \frac{\d}{\d w_{\gamma,n}}\left(\frac{\d \tilde F}{\d q_{\mu,i}}\Omega_{\nu,j;\beta,0}\right)\d_x^n \circ \sum_{s\geq 1} 
A^{\gamma\xi}_s \d_x^s \\ \notag &
= \mathrm{(I)} + \mathrm{(IV)} +
\sum_{n} \frac{\d}{\d w_{\gamma,n}}\left(\frac{\d \tilde F}{\d q_{\mu,i}}\d_x\Omega_{\nu,j;\beta,0}\right)\d_x^n \circ \sum_{s\geq 1} 
A^{\gamma\xi}_s \d_x^s  
\end{align}
Moreover,
\begin{align}\label{eq:d-x-1} &
\sum_{n} \frac{\d}{\d w_{\gamma,n}}\left(\frac{\d \tilde F}{\d q_{\mu,i}}\d_x\Omega_{\nu,j;\beta,0}\right)\d_x^n \circ \sum_{s\geq 1} 
A^{\gamma\xi}_s \d_x^s
\\ \notag &
= \sum_{n} \frac{\d \tilde F}{\d q_{\mu,i}} \frac{\d}{\d w_{\gamma,n}}\left(\d_x\Omega_{\nu,j;\beta,0}\right)\d_x^n \circ \sum_{s\geq 1} 
A^{\gamma\xi}_s \d_x^s  
\\ \notag &
+ \d_x\Omega_{\nu,j;\beta,0}\d_x^{-1} \circ \sum_{n} \frac{\d \Omega_{\mu,i;\un,0}}{\d w_{\gamma,n}}\d_x^n \circ \sum_{s\geq 1} 
A^{\gamma\xi}_s \d_x^s.
\end{align}
Using skewsymmetry of $A$, that is, $\sum_{s\geq 1} 
A^{\gamma\xi}_s \d_x^s=- \sum_{s\geq 1} (-\d_x)^s \circ
A^{\xi\gamma}_s$, the fact that $\sum_{s\geq 1} (-\d_x)^{s-1} \circ
A^{\gamma\xi}_s \delta_\gamma \Omega_{\mu,i;\un,0}=\Omega_{\mu,i-1;\xi,0}$, and the obvious equation $\d^{-1}\circ X \cdot \d_xY =XY - \d^{-1} \circ \d_xX\cdot Y$, we see that the last summand on the right hand side of Equation~\eqref{eq:d-x-1} is equal to
\begin{equation}\label{eq:bad-2-omega}
\mathrm{(VIII)} + \mathrm{(IX)} - \d_x  \Omega_{\nu,j;\beta,0} \d_x^{-1} \circ \Omega_{\xi,0;\mu,i-1}\circ \d_x.
\end{equation}

On the other hand,
\begin{equation}
\sum_{\gamma,s} A^{\beta\gamma}_s\d_x^s \circ \frac{\hbar}{2} \sum_{n} (-\d_x)^n \circ \frac{\d \Omega_{\mu,i;\nu,j;\xi,0}}{\d w_{\gamma,n}}\circ (-\d_x)
\end{equation}
is equal to the line (XI) in Equation~\eqref{eq:r-def-bracket}. This follows directly from the definition of the operator $\T_{\gamma,m}$.
Meanwhile,
\begin{align} \label{eq:right-h-0}
& \sum_{\gamma,s} A^{\beta\gamma}_s\d_x^s \circ \sum_{n} (-\d_x)^n \circ \frac{\d}{\d w_{\gamma,n}}\left(\frac{\d \tilde F}{\d q_{\mu,i}}\Omega_{\nu,j;\xi,0}\right)(-\d_x) \\ \notag &
=
\sum_{\gamma,s} A^{\beta\gamma}_s\d_x^s \circ \sum_{n} (-\d_x)^n \circ \frac{\d \Omega_{\mu,i;\un,0}}{\d w_{\gamma,n}}\d_x^{-1}\circ \Omega_{\nu,j;\xi,0}(-\d_x) \\ \notag &
+ \sum_{\gamma,s} A^{\beta\gamma}_s\d_x^s \circ \sum_{n} (-\d_x)^n \circ \frac{\d \tilde F}{\d q_{\mu,i}}\frac{\d \Omega_{\nu,j;\xi,0}}{\d w_{\gamma,n}}(-\d_x) 
\end{align}
For the first summand we try, using the Leibniz rule for $\d_x$, to move $\d_x^s\circ (-\d_x)^n$ to $\d_x^{-1}$. This gives
\begin{align}
& \sum_{\gamma,s} A^{\beta\gamma}_s\d_x^s \circ \sum_{n} (-\d_x)^n \circ \frac{\d \Omega_{\mu,i;\un,0}}{\d w_{\gamma,n}}\d_x^{-1}\circ \Omega_{\nu,j;\xi,0}(-\d_x) \\ \notag &
=\mathrm{(V)}+\mathrm{(VII)}+ \sum_{\gamma,s} A^{\beta\gamma}_s\d_x^s\delta_{\gamma,n}\Omega_{\mu,i;\un,0}\d_x^{-1}\circ \Omega_{\nu,j;\xi,0}(-\d_x) \\ \notag &
=\mathrm{(V)}+\mathrm{(VII)}+ \d_x\Omega_{\mu,i-1;\beta,0}\d_x^{-1}\circ \Omega_{\nu,j;\xi,0}(-\d_x).
\end{align}
Observe that the last summand here cancels the last summand in formula~\eqref{eq:bad-2-omega}. 

The second summand on the right hand side of Equation~\eqref{eq:right-h-0} we treat in the same way as the first summand: we apply the Leibniz rule to each factor in $\d_x^s\circ (-\d_x)^n$ trying to hit $\d \tilde F/\d q_{\mu,i}$. We have:
\begin{align}\label{eq:last-extr}
& \sum_{\gamma,s} A^{\beta\gamma}_s\d_x^s \circ \sum_{n} (-\d_x)^n \circ \frac{\d \tilde F}{\d q_{\mu,i}}\frac{\d \Omega_{\nu,j;\xi,0}}{\d w_{\gamma,n}}(-\d_x) \\ \notag &
= \mathrm{(VI)}+ \sum_{\gamma,s} A^{\beta\gamma}_s\d_x^s \circ \frac{\d \tilde F}{\d q_{\mu,i}}\sum_{n} (-\d_x)^n \circ \frac{\d \Omega_{\nu,j;\xi,0}}{\d w_{\gamma,n}}(-\d_x) \\ \notag &
= \mathrm{(VI)}+\mathrm{(III)}+\frac{\d \tilde F}{\d q_{\mu,i}}\sum_{\gamma,s} A^{\beta\gamma}_s\d_x^s \circ \sum_{n} (-\d_x)^n \circ \frac{\d \Omega_{\nu,j;\xi,0}}{\d w_{\gamma,n}}(-\d_x).
\end{align}

Using Lemma~\ref{lem:useful-lemma}, wee see that the last term in Equation~\eqref{eq:last-extr} and the first term on the right hand side of Equation~\eqref{eq:d-x-1} are exactly the extra terms~\eqref{eq:extra-terms} in the statement of the lemma.
\end{proof}

\subsection{Extra terms}

In this section, we deal with the terms coming from the third summand on the right hand side of Equation~\eqref{eq:def-quasi-miura}.
First of all, observe that the operator $\sum_n \d_x^n U \circ \d/\d u_n$ commutes with $\d_x$ (we use here notations of Lemma~\ref{lem:useful-lemma}, that is, $U$ is an arbitrary function and $u_n$ are either the variables $v_{\alpha,n}$ or $w_{\alpha,n}$ for some $\alpha$). This means that the two terms coming from the third summand on the right hand side of Equation~\eqref{eq:def-quasi-miura} can be written in terms of the operator 
$\sum_{s\geq 1} A^{\beta\xi}_s \d_x^s = \sum_{\alpha} L^\beta_\alpha\circ \d_x\circ \left(L^*\right)^\xi_\alpha$ as
\begin{equation}\label{eq:diff-terms}
-\sum_{\gamma,n} r[t].w_{\gamma,n}\sum_{\gamma,s} \frac{\d A^{\beta\xi}_s}{\d w_{\gamma,n}} \d_x^s.
\end{equation}

\begin{lemma}\label{lem:r-3} Formula~\eqref{eq:diff-terms} is equivalent to the sum of the lines~(II) and~(XII) on the right hand side of Equation~\eqref{eq:r-def-bracket} and the following extra term:
\begin{equation}\label{eq:extra-terms-2}
-
\sum_{l\geq 1}\sum_{i+j=l-1}(-1)^{i+1}(r_l)_{\mu\nu} 
\frac{\d\tilde F}{\d q_{\mu,i}}\sum_{\gamma,n,s}\frac{\d w_{\gamma,n}}{\d q_{\nu,j}} \frac{\d A^{\beta\xi}_s}{\d w_{\gamma,n}} \d_x^s.
\end{equation}
\end{lemma}

\begin{proof} The proof is a short straightforward computation. 
Indeed, 
\begin{align} &
-\sum_{\gamma,n,s} r[t].w_{\gamma,n} \frac{\d A^{\beta\xi}_s}{\d w_{\gamma,n}} \d_x^s \\ \notag &
=-\frac{\hbar}{2} 
\sum_{l\geq 1}\sum_{i+j=l-1}(-1)^{i+1}(r_l)_{\mu\nu} 
\sum_{\gamma,n,s}\d_x^{n+1}\Omega_{\mu,i;\nu,j;\gamma,0} \frac{\d A^{\beta\xi}_s}{\d w_{\gamma,n}} \d_x^s 
\\ \notag &
-\sum_{l\geq 1}\sum_{i+j=l-1}(-1)^{i+1}(r_l)_{\mu\nu} \sum_{\gamma,n,s}
\d_x^{n+1}\left(\frac{\d\tilde F}{\d q_{\mu,i}}\Omega_{\nu,j;\gamma,0}\right) \frac{\d A^{\beta\xi}_s}{\d w_{\gamma,n}} \d_x^s \\ \notag &
= \mathrm{(XII)}+\mathrm{(II)} -
\sum_{l\geq 1}\sum_{i+j=l-1}(-1)^{i+1}(r_l)_{\mu\nu} 
\frac{\d\tilde F}{\d q_{\mu,i}}\sum_{\gamma,n,s}\frac{\d w_{\gamma,n}}{\d q_{\nu,j}} \frac{\d A^{\beta\xi}_s}{\d w_{\gamma,n}} \d_x^s.
\end{align}
\end{proof}

\subsection{Proof of Theorem~\ref{thm:r-def-bracket}}\label{sec:proofthmr} The theorem is almost proved above, that is, we have already derived all summand on the right hand side of Equation~\eqref{eq:r-def-bracket}. But in the course of computations we got some additional terms given by formulas~\eqref{eq:extra-terms} and~\eqref{eq:extra-terms-2}. Let us recollect them here:
\begin{align} \label{eq:canterms}
&
\sum_{l\geq 1}\sum_{i+j=l-1}(-1)^{i+1}(r_l)_{\mu\nu} 
\frac{\d\tilde F}{\d q_{\mu,i}}\left( 
\sum_{\gamma,s,n} \frac{\d \left(\frac{\d w_\beta}{\d q_{\nu,j}}\right)}{\d w_{\gamma,n}}\d_x^n \circ  
A^{\gamma\xi}_s \d_x^s  \right.
\\ \notag &
\left. + 
 \sum_{\gamma,s,n} A^{\beta\gamma}_s \d_x^s \circ (-\d_x)^n\circ 
 \frac{\d \left(\frac{\d w_\xi}{\d q_{\nu,j}}\right)}{\d w_{\gamma,n}}
 -\sum_{\gamma,n,s}\frac{\d w_{\gamma,n}}{\d q_{\nu,j}} \frac{\d A^{\beta\xi}_s}{\d w_{\gamma,n}} \d_x^s\right) .
\end{align}

Theorem~\ref{thm:r-def-bracket} follows from the following lemma.
\begin{lemma} The formula~\eqref{eq:canterms} is equal to zero. \end{lemma}

\begin{proof} We prove below (see Lemma~\ref{lem:int2}) that 
\begin{equation}\label{eq:delta-int}
-L^{\beta}_\alpha\circ \left( \sum_{p}\frac{\d \left(\frac{\d v_{\alpha,0}}{\d q_{\nu,j}}\right)}{\d v_{\gamma,p}}\circ \d_x^{p}\circ \d_x
+ \d_x\circ \sum_{p} (-\d_x)^p \circ \frac{\d \left(\frac{\d v_{\gamma,0}}{\d q_{\nu,j}}\right)}{\d v_{\alpha,p}}\right)\circ \left(L^*\right)^{\xi}_\gamma=0
\end{equation}
Observe that the sum of these two terms and the three summands in the brackets in formula~\eqref{eq:canterms} gives an expression of exactly the same type as considered in Lemmas~\ref{lem:r-1},~\ref{lem:r-2}, and~\ref{lem:r-3}, with the difference that the operator $r[t].$ is replaced by $\d/\d q_{\nu,j}$. This means that the sum of all these five summands is equal to 
\begin{equation} \label{eq:delta-l}
\sum_\alpha \Delta L^\beta_\alpha \circ \d_x \circ \left(L^*\right)^{\xi}_\alpha + L^\beta_\alpha \circ \d_x \circ 
\left(\Delta L^*\right)^{\xi}_\alpha,
\end{equation}
where $\Delta L^\beta_\alpha$ is equal to the right hand side of Equation~\eqref{eq:def-quasi-miura} with $r[t].$ replaced by $\d/\d q_{\nu,j}$, that is,
\begin{align}
\sum_s \Delta L^\beta_{\alpha,s}\d_x^s & =  
\sum_{\gamma,n} \frac{\d\left(\frac{\d w_{\beta}}{\d q_{\nu,j}}\right)}{\d w_{\gamma,n}}\d_x^n
\circ
\sum_{s=0}^\infty L^\gamma_{\alpha,s}\d_x^s \\ \notag
& -
\sum_{\gamma,s}^\infty L^\beta_{\gamma,s}\d_x^s
\circ
\sum_{n} \frac{\d\left(\frac{\d v_{\gamma}}{\d q_{\nu,j}}\right)}{\d v_{\alpha,n}}\d_x^n \\ \notag
& - \sum_{\gamma,n,s} \frac{\d w_{\gamma,n}}{\d q_{\nu,j}}\frac{\d L^\beta_{\alpha,s}}{\d w_{\gamma,n}}\d_x^s.
\end{align}
Lemma~\ref{lem:U-q} used for $U=1$ and $q=q_{\nu,j}$ implies that $\Delta L^\beta_\alpha$ is equal to zero. Therefore, formula~\eqref{eq:delta-l} is equal to zero. Then, since formula~\eqref{eq:delta-int} is shown below to be equal to zero, we obtain that the sum of the three summands in the brackets in formula~\eqref{eq:canterms} is equal to zero, and, therefore, the whole formula~\eqref{eq:canterms} vanishes as well. 
\end{proof}

The still missing ingredient of our computations is the following lemma that implies that the formula~\eqref{eq:delta-int} is equal to zero.
\begin{lemma}\label{lem:int2}
We have:
\begin{equation}\label{eq:int2}
\sum_{p}\frac{\d \left( \frac{\d v_{\alpha,0}}{\d q_{\nu,j}}\right)}{\d v_{\beta,p}}\circ \d_x^{p}\circ \d_x
+ \d_x\circ \sum_{p} (-\d_x)^p \circ \frac{\d \left(\frac{\d v_{\beta,0}}{\d q_{\nu,j}}\right)}{\d v_{\alpha,p}} =0.
\end{equation}
\end{lemma}

\begin{proof} The proof is analogous to the proof of Lemma~\ref{lem:r-1}. Indeed, using Lemma~\ref{lem:useful-lemma}, we see that
\begin{equation}
\sum_{p}\frac{\d \left( \frac{\d v_{\alpha,0}}{\d q_{\nu,j}}\right)}{\d v_{\beta,p}}\circ \d_x^{p}
= \d_x\circ \sum_{p}\frac{\d \Omega^{[0]}_{\alpha,0;\nu,j}}{\d v_{\beta,p}}\circ \d_x^{p}
= \d_x \circ \frac{\d \Omega^{[0]}_{\alpha,0;\nu,j}}{\d v_{\beta,0}}.
\end{equation}
Meanwhile, in the same way,
\begin{equation}
\sum_{p} (-\d_x)^p \circ \frac{\d \left(\frac{\d v_{\beta,0}}{\d q_{\nu,j}}\right)}{\d v_{\alpha,p}}
=\frac{\d \Omega^{[0]}_{\beta,0;\mu,i}}{\d v_{\alpha,0}}\circ (-\d_x).
\end{equation}
Therefore, the left hand side of Equation~\eqref{eq:int2} is equal to
\begin{equation}
\d_x\circ \left( \frac{\d \Omega^{[0]}_{\alpha,0;\nu,j}}{\d v_{\beta,0}} - \frac{\d \Omega^{[0]}_{\beta,0;\nu,j}}{\d v_{\alpha,0}} \right) \circ \d_x,
\end{equation}
which is equal to zero because of Equation~\eqref{eq:triomega}.
\end{proof}


\subsection{The $s$-deformation formula} In this Section we prove Equation~\eqref{eq:s-deformation-bracket} using Theorem~\ref{thm:s-def-quasi-miura}.

\begin{theorem}\label{thm:s-def-bracket} Let $L$ be the operator of quasi-Miura tranformation discussed in Section~\ref{sec:r-def-L}. Then the 
deformation formula 
\begin{equation}\label{eq:s-def-bracket-L}
\sum_\alpha s[w].L^\beta_\alpha \circ \d_x \circ \left(L^*\right)^\xi_\alpha 
+ \sum_\alpha L^\beta_\alpha \circ \d_x \circ \left(s[w].L^*\right)^\xi_\alpha  
\end{equation}
is equal to the right hand side of Equation~\eqref{eq:s-deformation-bracket}.
\end{theorem}

\begin{proof} Theorem~\ref{thm:s-def-quasi-miura} implies that we can replace $s[w].$ in the formula~\eqref{eq:s-def-bracket-L} by $-\sum_\nu(s_1)_{\nu,\un}\d/\d w_{\nu,0}$. Since the operator $\d/\d w_{\nu,0}$ commutes with $\d_x$, we see that the formula~\eqref{eq:s-def-bracket-L} is equal to
\begin{equation}
-\sum_\nu(s_1)_{\nu,\un}\frac{\d}{\d w_{\nu,0}}\left(L^\beta_\alpha \circ \d_x \circ \left(L^*\right)^\xi_\alpha \right),
\end{equation}
which is exactly the right hand side of Equation~\eqref{eq:s-deformation-bracket}.
\end{proof}




\end{document}